\documentclass[nonacm,screen,manuscript]{acmart}
\pdfoutput=1

\usepackage{float}
\usepackage{color}
\usepackage{makecell}
\usepackage[noend]{algpseudocode}
\usepackage{algorithm}
\usepackage{thm-restate}

\newcounter{todocounter}
\newcommand{\todonum}{\stepcounter{todocounter}{(\thetodocounter)}}

\def\shownotes{1}   
\ifnum\shownotes=1
\newcommand{\authnote}[2]{{ $\ll$\textsf{\footnotesize \todonum\  #1 notes:  #2}$\gg$}}
\else
\newcommand{\authnote}[2]{}
\fi

\newcommand{\gs}[1]{{\color{magenta}\authnote{Gilad}{#1}}}
\newcommand{\ia}[1]{{\color{red}\authnote{Ittai}{#1}}}

\newcommand{\ignore}[1]{}
\newcommand{\hide}[1]{}

\acmBooktitle{}{}


\begin{document}

\author{Ittai Abraham}
\affiliation{VMWare Research}
\author{Gilad Stern}
\affiliation{Hebrew University}

\algdef{SN}{AsLongAs}{EndAsLongAs}[1]{\textbf{as long as} #1, \textbf{run}}
\algdef{SN}{AtTime}{EndAtTime}[1]{\textbf{at time} #1 \textbf{do}}
\algdef{SN}{Upon}{EndUpon}[1]{\textbf{upon} #1, \textbf{do}}
\algdef{SN}{Wait}{EndWait}[1]{\textbf{wait until} #1, \textbf{then do}}
\newcommand{\viewLength}{\ensuremath{11\Delta} }


\title{Information Theoretic HotStuff}

\begin{abstract}
This work presents Information Theoretic HotStuff (IT-HS), a new optimally resilient protocol for solving Byzantine Agreement in partial synchrony with information theoretic security guarantees. In particular, IT-HS does not depend on any PKI or common setup assumptions and is resilient to computationally unbounded adversaries.
IT-HS is based on the Primary-Backup view-based paradigm.  In IT-HS, in each view, and in each view change, each party sends only a constant number of words to every other party.
This yields an $O(n^2)$ word and message complexity in each view.
In addition, IT-HS requires just $O(1)$ persistent local storage and $O(n)$ transient local storage.
Finally,  like all Primary-Backup view-based protocols in partial synchrony, after the system becomes synchronous, all nonfaulty parties decide on a value in the first view a nonfaulty leader is chosen.
Moreover, like PBFT and HotStuff, IT-HS is optimistically responsive: with a nonfaulty leader, parties decide as quickly as the network allows them to do so, without regard for the known upper bound on network delay.
Our work improves in multiple dimensions upon the information theoretic version of PBFT presented by Miguel Castro, and can be seen as an information theoretic variant of the HotStuff paradigm.
\end{abstract}

\maketitle

\section{Introduction}
This work assumes the model of Castro and Liskov's PBFT protocol \cite{cachin2010visitpaxos,castro2001PBFT,castro1999practical}. 
In particular we deal with the task of Byzantine Agreement in a partially synchronous network.
The setting of partial synchrony was proposed by Dwork, Lynch, and Stockmeyer \cite{DLS} and studied extensively since.
In this model, the network starts off as an asynchronous network and at some unknown time becomes synchronous with a known delay $\Delta$ on message arrival.
This time is known as the Global Stabilization Time, or GST in short.
This model turns out to be a useful one,  managing to capture some of the behaviour of real-world networks.
As in PBFT, our goal in this work is to reduce the use of cryptographic tools that require a computationally bounded adversary as much as possible. 
Much like PBFT, our algorithm is \textit{information theoretically secure}. 
Formally, as in PBFT \cite{cachin2010visitpaxos,castro2001PBFT, castro1999practical}, our protocol is secure against adversaries that are not computationally bounded under the assumption that there exist authenticated channels that can be made secure against such adversaries. 
For example, authenticated channels can be obtained via a setup of one time pads or via Quantum key exchange \cite{Bennett1992}.

There are several good reasons to design protocols in the information theoretic security setting.
First, from a theoretical perspective we are interested in  minimizing the assumptions. 
Fewer assumptions often tend to add clarity and conceptual simplicity.
Secondly, adding public-key cryptography primitives adds a performance overhead and increases the code-base attack surface, whereas computations in the information-theoretic setting are quick and often amount to simple memory management and counting.
Finally, protocols in this setting are more ``future-proof''.
Such protocol are more resilient to breaking certain cryptographic assumptions and to major technological disruptions in the field.

The PBFT variants that use a PKI and digital signatures can easily use bounded storage at each party (per active slot).
One of the challenges of the PBFT protocol when only authenticated channels (no signatures) are used is that obtaining bounded storage is not immediate. Indeed all the peer reviewed papers that we are aware of obtain unbounded solutions \cite{cachin2010visitpaxos,pbft_tocs}. Castro's thesis \cite{castro2001PBFT} does include a bounded storage solution, however to the best of our knowledge this result was not published in a peer reviewed venue, and its complexity does rely on cryptographic hash functions.

\subsection{Main result}
Our main result is \textit{Information Theoretic HotStuff} (IT-HS), a protocol solving the task of Byzantine Agreement in partial synchrony with information theoretic security using bounded storage that sends messages whose maximal size is $O(1)$ words (both during a view and during a view change).
The protocol is resilient to any number of Byzantine parties $f$ such that $n>3f$, making it optimally resilient.
In the protocol, there are several virtual rounds called views, and each one has a leader, called a primary.
This is a common paradigm for solving Byzantine agreement, famously used in the Paxos protocol \cite{lamport2001paxos} and in later iterations on those ideas such as PBFT \cite{cachin2010visitpaxos,castro2001PBFT,zyzzyva} and more recent protocols in the Blockchain era \cite{buchman2016thesis, buchman2018tendermint, buterin2017casper,gueta2018sbft,abraham2019hotstuff}.
We use a standard measure of storage called a \textit{word} and assume a word can contain enough information to store any command, identifier,  or counter. 
Formally, this means that much like in all previous systems and protocols, our counters, identifiers, and views are bounded (by say 256 bits).
In IT-HS, in each view and in each view change, each party sends just a constant number of words and messages to each other party, making the total word and message complexity $O(n^2)$ in each view and in each view change.
As far as we know, this is the best known communication complexity and word complexity for information theoretic protocols of this kind (see table below for comparison).
In addition, all parties require $O(n)$ space throughout the protocol, out of which only $O(1)$ space needs to be persistent, crash-resistant memory.
Clearly at least $O(1)$ persistent memory is required, because otherwise a decided upon value can be ``forgotten'' by all parties if they crash and reboot.
As far as we know, $O(n)$ transient space complexity is the best known result.
In the shared memory model, a lower bound of $\Omega(n)$ registers exists \cite{gelashvili2015optimal}, suggesting that the total amount of persistent memory in the system is optimal. 

In IT-HS, all nonfaulty parties are guaranteed to decide on a value and terminate during the first view a nonfaulty party is chosen as primary after GST, if they haven't done so earlier.
This is the asymptotically optimal convergence for such protocols: For deterministic leader rotation this implies $O(f)$ rounds after GST. If we assume that  parties have access to a randomized leader-election beacon, then this implies $O(1)$ expected rounds after GST.
Furthermore, like PBFT and HotStuff, IT-HS is \textit{optimistically responsive}. If the network delay is actually $\delta=o\left(\Delta\right)$, all nonfaulty parties terminate in $O\left(\delta\right)$ time instead of in $O\left(\Delta\right)$ time.
IT-HS uses an \emph{asymptotically} optimal (constant) number of rounds given a nonfaulty primary and after the network becomes synchronous.

The most relevant related works for IT-HS are the PBFT protocol variants \cite{cachin2010visitpaxos,castro2001PBFT,pbft_tocs,castro1999practical} and the HotStuff protocol variants \cite{abraham2019hotstuff}. The following table provides a comparison between them.
\begin{center}
\begin{tabular}{lccc} 
\hline
  & Assumptions & Persistent storage & \makecell{Maximum size of\\ message (in words)}\\ 
\hline\hline 
PBFT (OSDI) \cite{castro1999practical} & {\color{red} PKI}& $\Omega(n)$ & $O(n)$ \\ 
\hline
PBFT (TOCS) \cite{pbft_tocs}& \makecell{Authenticated Channels,\\ {\color{red} Cryptographic Hash}} & \makecell{$\Omega(n)$ per view\\ {\color{red} (unbounded)}} &  \makecell{$\Omega(n)$ per view\\ {\color{red} (unbounded)}} \\ 
\hline
PBFT (Thesis) \cite{castro2001PBFT} & \makecell{Authenticated Channels,\\ {\color{red} Cryptographic Hash}} & $O(1)$ & $O(n)$\\ 
\hline
YAVP (Cachin) \cite{cachin2010visitpaxos} & \makecell{Authenticated Channels, \\ {\color{red} Cryptographic Hash}} &  
\makecell{$\Omega(n)$ per view\\ {\color{red} (unbounded)}} &  \makecell{$\Omega(n)$ per view\\ {\color{red} (unbounded)}} \\ 
\hline
\makecell{HotStuff \\ (authenticators)  \cite{abraham2019hotstuff}}& {\color{red} PKI} & $O(n)$ & $O(n)$ \\ 
\hline
\makecell{HotStuff \\ (threshold sig) \cite{abraham2019hotstuff}} & \makecell{{\color{red} DKG: Threshold}\\ {\color{red} signature setup}} & $O(1)$ & \makecell{$O(1)$\\ ({\color{red} threshold sig}) }  \\ 
\hline
\textbf{IT-HS} (this work) & Authenticated Channels  &  $O(1)$ & $O(1)$  \\ 
\hline
\end{tabular}
\end{center}

As mentioned earlier, all previous peer-reviewed works in the information theoretic setting require at least $\Omega(n \cdot v)$ words of storage, where $v$ is the view number. Since the view number can grow arbitrarily large, the persistent storage requirement is unbounded. 
The only work we know of that achieves comparable asymptotic performance relies on the relatively strong cryptographic assumption of threshold signatures.

We note that IT-HS does not only use fewer assumptions (does not use any cryptographic hash function), it also obtains the asymptotically optimal $O(1)$ word bound on the maximal message size. 
All other protocols require at least $\Omega(n)$ size messages to be sent during view change by the primary (except for Hotstuff when using a Distributed Key Generation setup and threshold signatures).

Compared to PBFT, our work can be seen as addressing the open problem left in the PBFT journal version (which uses unbounded space and cryptographic hash functions) and is an improvement of the non peer-reviewed PBFT thesis work (which still uses cryptographic hash functions). IT-HS obtains the same $O(1)$ persistent space,  and manages to reduce the maximum message size from $O(n)$ (in the PBFT view change) to the asymptotically optimal $O(1)$ maximum message size and requires no cryptographic hash functions.

Relative to HotStuff, our work shows that without any PKI (public key infrastructure) or DKG (distributed key generation) assumptions and without any cryptographic setup ceremony, constant size messages and constant size persistent storage are possible! 
We do note that IT-HS requires $O(n^2)$ messages and words per view, while the Hostuff version with a DKG setup that uses threshold signatures requires just $O(n)$ messages and words per view. 
On the other hand, HS-IT requires no cryptographic setup ceremony and no computational assumptions other than pairwise authenticated channels. 
Like HS-IT, all other protocols that do not use threshold signatures (even those that require a PKI) use $\Omega(n^2)$ words per view.

\subsubsection*{Our contributions}

\begin{enumerate}
    \item Unlike previous solutions which used cryptographic hash functions and required $O(n)$ sized messages, We provide the first information theoretic primary backup protocol where all messages have size $O(1)$ and storage is bounded to size $O(1)$.
    \item We manage to reduce the size of the view change messages to a constant by adapting the HosStuff paradigm without using any cryptographic primitives. 
    We introduce an information theoretic technique for one-transferable signatures to maintain bounded space and adopt the view change protocol accordingly.
    \item Without using any cryptographic primitives, we obtain a protocol that requires just a constant amount of persistent storage. 
    We use information theoretic techniques that require storing just the last two events from each message type. 
\end{enumerate}

\subsection{Main Techniques}
As the name might suggest, IT-HS is inspired by the Tendermint, Casper, and HotStuff protocols \cite{buchman2016thesis, buchman2018tendermint, buterin2017casper, abraham2019hotstuff} and adapts them to the information theoretic setting.
We show how to adapt the  \textit{lock and key} mechanism which was suggested in HotStuff \cite{abraham2019hotstuff} and made explicit in \cite{abraham2019vaba}, to the information theoretic setting while maintaining just  $O(1)$ persistent storage.
In a basic \textit{locking} mechanism \cite{buchman2016thesis,DLS}, before nonfaulty parties decide on a value, they set a ``lock'' that doesn't allow them to respond to primaries suggesting values from older views.
Then, before deciding on a value, nonfaulty parties require a proof that enough parties are locked on the current view.
This ensures that if some value is decided upon, there will be a large number of nonfaulty parties that won't be willing to receive messages from older views, and thus this will remain the only viable value in the system.

The challenge with the locking mechanism is that the adversary can cause nonfaulty locked parties to block nonfaulty primaries, unless the primary waits for all nonfaulty parties to respond.
To overcome this, an additional round is added so that a nonfaulty locked party guarantees that there is a sufficient number of nonfaulty parties with a \textit{key}. 
When a new primary is chosen, it waits for just $n-f$ parties to send their highest keys, and uses the highest one it receives.

The challenge with using a key is verifying its authenticity.
In the cryptographic setting, this  is easily done using signatures.
In the information theoretic setting, verification is more challenging. One approach is using Bracha's Broadcast \cite{bracha1987broadcast} in order to prove that the key received by the primary will also be accepted by the other parties.
Since there is no indication of termination in Bracha's Broadcast, there is a need to maintain an unbounded number of broadcast instances (one for each view). Using such techniques requires an unbounded amount of space. 

To overcome this challenge with bounded space, we propose a novel approach of using  \textit{one-hop transferable} proofs.
If before moving to the next round, a nonfaulty party hears from $n-f$ parties, then it knows it heard from at least $f+1$ nonfaulty parties.
This means that once the system becomes synchronous, every party will hear from those $f+1$ parties and know that at least one of them is nonfaulty.
We use this type of ``one-hop transferable proof'' twice so we have 3 key messages instead of one, each proving that the next key (or lock) is correct, and that this fact can be proven to other parties, thereby ensuring liveness.

%

In order to send just a constant number of words, we send just the last two times that the value of the $key$ was updated.
If the final update to $key$ happened after a lock was set, and its value is different than the lock's value, then the lock is safe to open.
Otherwise, if the older of the two updates was after the lock's view then at least in one of those times it was updated to a value other than the lock's value, and thus the lock is also safe to open. 
Using this idea, parties can also prove to a primary that a $key3$ suggestion is safe.
In this case, the parties either show a later view in which the same value was set for $key2$, or two later views in which the value of $key2$ was updated.
This proof shows that any previous lock either has the same value as $key3$, or can be opened safely regardless of its value.
The idea of storing just two lock values appears in Castro's Thesis \cite{castro2001PBFT}, we significantly extend this technique to use  our novel  \textit{one-hop transferable}  information theoretic ``signatures'' combined with the HotStuff keys-lock approach.


\subsection{Protocol Overview}
Much like all primary backup protocols, each view of IT-HS consists of a constant number of rounds. 
Each party waits to receive $n-f$ round $i$ messages before it sends a round $i+1$ message (in some rounds there are additional checks). 
Much like PBFT, each round involves an all-to-all message sending format. %
Throughout the protocol, parties may set a lock for a given view and value.
This lock indicates that any proposal for a different $view, value$ pair should not be accepted without ample proof that another value reached advanced stages in a later view.
In order to provide that proof, the parties send a $proof$ message that helps convince parties with locks to accept messages about a different value if appropriate.

The rounds of IT-HS for a given view can be partitioned into 4 parts:
\begin{enumerate}
\item View Change: parties first send a $request$ message, indicating that they started the view. 
Once parties hear the $request$ message sent by the primary, they respond with their current suggestion for a value to propose, as well as the view in which this suggestion originated, and additional data which will help validate all nonfaulty parties' suggestions (proofs). 
After receiving those suggestions, the primary checks whether each suggestion is valid, and once it sees $n-f$ valid suggestions, it sends a $propose$ message for the one that originated in the most recent view.

\item Propose message round:  this is where a party checks a proposal relative to its lock. 
Each party checks if it's locked on the same value as the one proposed, or convinced to override its lock by $f+1$ proof messages. If that is the case, it responds by sending an $echo$ message.

\item Key message rounds: this is where a key is created that can be later used to unlock parties.
After receiving $n-f$ $echo$ messages with the same value, parties send a $key1$ message with that value.
After receiving $n-f$ $key1$ messages with the same value they send a $key2$ message. 
After receiving $n-f$ $key2$ messages with the same value they send a $key3$ message. 
We use these three rounds in order to obtain transferable information theoretic signatures on the key message.

\item Lock and commit rounds: After receiving $n-f$ $key3$ messages with the same value they \textit{lock} on it and send a $lock$ message. 
After receiving $n-f$ $lock$ messages with the same value they \textit{commit} and send a $done$ message.  
\end{enumerate}

Before sending a $key1$ message, the local $key1$, $key1\_val$ and $prev\_key1$ fields are updated.
These fields contain the last view in which a $key1$ message was sent, its value, and the last view a $key1$ message was sent with a different value.
Similar updates take place for the other $key$ fields and the $lock$ fields.
The $echo$, $lock$ and various $key$ messages are tagged with the current view, while the $done$ message is a protocol-wide message and isn't related to a specific view.
Similarly to the mechanism in Bracha Broadcast \cite{bracha1987broadcast}, after receiving $f+1$ $done$ messages, the message is echoed, and after receiving $n-f$ messages it is accepted and the parties decide and terminate.
If a party sees that this view takes more than the expected time, it sends an $abort$ message for the view.
The same $f+1$ threshold for echoing the $abort$ message and $n-f$ threshold for moving to the next view are implemented in order to achieve the same properties.
In order to avoid buffering $request$ and $abort$ messages, only the messages with the highest view $v$ are actually stored and are understood as a $request$ or $abort$ message for any view up to $v$.

\section{Byzantine Agreement in Partial Synchrony}
This section deals with the task of Byzantine Agreement in a partially synchronous system.
In this model, there exist $n$ parties who have local clocks and authenticated point-to-point channels to every other party.
The system starts off fully asynchronous: the clocks are not synchronized, and every message can be delayed any finite amount of time before reaching its recipient.
At some point in time, the system becomes fully synchronous: the clocks become synchronized, and every message (including the ones previously sent) arrives in $\Delta$ time at most, for some commonly known $\Delta$.
It is important to note that even though it is guaranteed that the system eventually becomes synchronous, the parties do not know when it is going to happen, or even if it has already happened.
The point in time in which the system becomes synchronous is called the Global Stabilization Time, or GST in short.
In the setting of a Byzantine adversary, the adversary can control up to $f$ parties, making them arbitrarily deviate from the protocol.
In general, throughout this work assume that $f<\frac{n}{3}$.

\begin{definition}
A Byzantine Agreement protocol in partial synchrony has the following properties:
\begin{itemize}
    \item \textbf{Termination.} If all nonfaulty parties participate in the protocol, they all eventually decide on a value and terminate.
    \item \textbf{Correctness.} If two nonfaulty parties decide on values $val,val'$, then $val=val'$.
    \item \textbf{Validity.} If all parties are nonfaulty and they all have the same input $val$, then every nonfaulty party that decides on a value does so with the value $val$.
\end{itemize}
\end{definition}

We note that if we assume the parties have access to an external validity function, as described in \cite{cachin2001validity}, this protocol can be easily adjusted to have external validity.
In this setting, the external validity function defines which values are ``valid'', and all nonfaulty parties are required to output a valid value.
The only adjustment needed is for parties to also check if a value is valid before sending an $echo$ message.

\begin{algorithm}[t]
\caption{IT-HS}\label{alg:BA}
Code for party $i$ with input $x_i$:
\begin{algorithmic}[1]
\State $lock\gets 0,lock\_val\gets x_i$
\State $key3\gets 0$, $key3\_val\gets x_i$
\State $key2\gets 0, key2\_val\gets x_i,prev\_key2\gets -1$
\State $key1\gets 0, key1\_val\gets x_i,prev\_key1\gets-1$
\State $view\gets 0$
\State $\forall j\in\left[n\right]\ highest\_request\left[j\right]\gets 0$
\State continually run $check\_progress()$ in the background
\While {true} \Comment{memory from last $process\_messages$ and $view\_change$ calls is freed}
    \State $cur\_view\gets view$
    \AsLongAs{$cur\_view= view$}
        \AtTime {$cur\_time() + \viewLength$}
            \State \label{line:abortMessage} send an $\left<abort,view\right>$ message to all parties
        \EndAtTime
        \State ignore messages from other views, other than $abort$, $done$ and $request$ messages
        \State $primary\gets (view\mod n) + 1$
        \State continually run $process\_messages(view)$ in the background
        \State $view\_change(view,primary)$
    \EndAsLongAs
\EndWhile
\end{algorithmic}
\end{algorithm}

\begin{algorithm}[t]
\caption{view\_change(view,primary)}
Code for party $i$:
\begin{algorithmic}[1]
\State send $\left<request,view\right>$ to all parties $j\in\left[n\right]$
\Upon{$highet\_request\left[primary\right]=view$}
    \State send $\left<suggest,key3,key3\_val,key2,key2\_val,prev\_key2,view\right>$ to $primary$
\EndUpon
\State $send\_all\_upon\_join(\left<proof,key1,key1\_val,prev\_key1,view\right>)$
\If {$primary=i$}
    \State $suggestions\gets\emptyset$
    \State $key2\_proofs\gets\emptyset$
    \Upon{receiving the first $\left<suggest,k3,v3,k2,v2,pk2,view\right>$ message from $j$}
        \If{$pk2<k2<view$}
            \State add $\left(k2,v2,pk2\right)$ to $key2\_proofs$
        \EndIf
        \If{$k3=0$}
            \State add $\left(k3,v3\right)$ to $suggestions$
        \ElsIf {$k3<view$}
            \Upon{$accept\_key\left(k3,v3,key2\_proofs\right)=true$}
                \State add $\left(k3,v3\right)$ to $suggestions$
            \EndUpon
        \EndIf
    \EndUpon
    \Wait{$\left|suggestions\right|\geq n-f$}
        \State let $(k,v)\in suggestions$ be some tuple such that $\forall \left(k',v'\right)\in suggestions\ k'\leq k$ 
        \State $send\_all\_upon\_join(\left<propose,k,v,view\right>)$
    \EndWait
\EndIf
\end{algorithmic}
\end{algorithm}

\begin{algorithm}[t]
\caption{accept\_key(key,value,proofs)}\label{alg:checkKey2}
\begin{algorithmic}[1]
\State $supporting\gets 0$
\ForAll{$\left(k,v,pk\right)\in proofs$}
    \If{$key\leq pk$}
        \State $supporting\gets supporting+1$
    \ElsIf{$key\leq k\land value=v$}
        \State $supporting\gets supporting+1$
    \EndIf
\EndFor
\If{$supporting\geq f+1$}
    \State return $true$
\Else
    \State return $false$
\EndIf
\end{algorithmic}
\end{algorithm}

\begin{algorithm}[t]
\caption{send\_all\_upon\_join(message)}
Code for party $i$:
\begin{algorithmic}[1]
\ForAll{parties $j\in\left[n\right]$}
    \Upon{$highest\_request\left[j\right]=view$}
        \State send $message$ to party $j$
    \EndUpon
\EndFor
\end{algorithmic}
\end{algorithm}

\begin{algorithm}[t]
\caption{check\_progress()}
Code for party $i$:
\begin{algorithmic}[1]
\State $\forall j\in\left[n\right]\  highest\_abort\left[j\right]\gets 0$
\Upon{receiving a $\left<request,v\right>$ message from party $j$}
    \If{$highest\_request\left[j\right]<v$}
        \State $highest\_request[j]\gets v$
    \EndIf
\EndUpon
\Upon{receiving a $\left<done,val\right>$ message from $f+1$ parties with the same $val$}
    \If{no $done$ message has been previously sent}
        \State \label{line:doneEcho} send $\left<done,val\right>$ to every party $j\in\left[n\right]$
    \EndIf
\EndUpon

\Upon{receiving a $\left<done,val\right>$ message from $n-f$ parties with the same $val$}
    \State \textbf{decide} val and \textbf{terminate}
\EndUpon

\Upon{receiving an $\left<abort, v\right>$ message from party $j$}
    \If{$highest\_abort\left[j\right]<v$}
        \State $highest\_abort\left[j\right]\gets v$
        \State let $u$ be the $f+1$'th largest value in $highest\_abort$
        \If{$u> highest\_abort\left[i\right]$}
            \State \label{line:abortEcho} send $\left<abort,u\right>$ to every party $j\in\left[n\right]$
            \State $highest\_abort\left[i\right]\gets u$
        \EndIf
        \State let $w$ be the $n-f$'th largest value in $highest\_abort$
        \If{$w\geq view$}
            \State $view\gets w+1$
        \EndIf
    \EndIf
\EndUpon
\end{algorithmic}
\end{algorithm}

\begin{algorithm}[t]
\caption{process\_messages(view)}
Code for party $i$:
\begin{algorithmic}[1]
\State $proofs\gets\emptyset$
\Upon{receiving the first $\left<proof,k1,v1,pk1,view\right>$ message from $j$}
    \If{$view>k1>pk1$}
        \State add $\left(k1,v1,pk1\right)$ to $proofs$
    \EndIf
\EndUpon
\Upon{receiving the first $\left<propose,key,val,view\right>$ message from $primary$}
    \If{$lock=0\lor val=lock\_val$}
            \State $send\_all\_upon\_join(\left<echo,val,view\right>)$
    \ElsIf{$view>key\geq lock$}
        \Upon{$open\_lock\left(proofs\right)=true$}
            \State $send\_all\_upon\_join(\left<echo,val,view\right>)$
        \EndUpon
    \EndIf
\EndUpon
\Upon{receiving an $\left<echo,val,view\right>$ message from $n-f$ parties with the same $val$}
    \State $send\_all\_upon\_join(\left<key1,val,view\right>)$
    \If{$key1\_val\neq val$}
        \State $prev\_key1\gets key1,key1\_val\gets val$
    \EndIf
    \State $key1\gets view$
\EndUpon
\Upon{receiving a $\left<key1,val,view\right>$ message from $n-f$ parties with the same $val$}
    \State $send\_all\_upon\_join(\left<key2,val,view\right>)$
    \If{$key2\_val\neq val$}
        \State $prev\_key2\gets key2,key2\_val\gets val$
    \EndIf
    \State $key2\gets view$
\EndUpon
\Upon{receiving a $\left<key2,val,view\right>$ message from $n-f$ parties with the same $val$}
    \State $send\_all\_upon\_join(\left<key3,val,view\right>)$ 
    \State $key3\gets view,key3\_val\gets val$
\EndUpon
\Upon{receiving a $\left<key3,val,view\right>$ message from $n-f$ parties with the same $val$}
    \State $send\_all\_upon\_join(\left<lock,val,view\right>)$ to every party $j\in\left[n\right]$
    \State $lock\gets view,lock\_val\gets val$
\EndUpon
\Upon{receiving a $\left<lock,val,view\right>$ message from $n-f$ parties with the same $val$}
    \If{no $done$ message has been previously sent}
        \State \label{line:doneMessage} send $\left<done,val\right>$ to every party $j\in\left[n\right]$
    \EndIf
\EndUpon
\end{algorithmic}
\end{algorithm}

\begin{algorithm}[t]
\caption{open\_lock(proofs)}\label{alg:checkKey1}
Code for party $i$:
\begin{algorithmic}[1]
\State $supporting\gets 0$
\ForAll{$\left(k,v,pk\right)\in proofs$}
    \If{$lock\leq pk$}
        \State $supporting\gets supporting+1$
    \ElsIf{$lock\leq k\land v\neq lock\_val$}
        \State $supporting\gets supporting+1$
    \EndIf
\EndFor
\If{$supporting\geq f+1$}
    \State return $true$
\Else
    \State return $false$
\EndIf
\end{algorithmic}
\end{algorithm}

The main goal of this section is to show that Algorithm~\ref{alg:BA} is a Byzantine Agreement protocol in partial synchrony resilient to $f<\frac{n}{3}$ Byzantine parties.
For ease of discussion, a party is said to perform an action ``in view $v$'' if when it performed the action its local $view$ variable equaled $v$.
In addition, we define the notion of messages ``supporting'' a key or opening a lock:
\begin{definition}
    A $suggest$ message is said to support the pair $key3,key3\_val$, if its $key2$, $key2\_val$, and $prev\_key2$  fields are ones for which at least one of the conditions in the loop of Algorithm~\ref{alg:checkKey2} is true.
    
    A $proof$ message is said to support opening the pair $lock,lock\_val$ if its $key1$, $key1\_val$, and $prev\_key1$ fields are ones for which at least one of the conditions in the loop of Algorithm~\ref{alg:checkKey1} is true.
\end{definition}

Before proving that Algorithm~\ref{alg:BA} is a Byzantine Agreement protocol in partial synchrony, we prove several lemmas. 
The lemmas can be classified into two types: safety lemmas and liveness lemmas.
The safety lemmas show that if a nonfaulty party decides on some value, no nonfaulty party decides on a different value.
This is achieved by the locking mechanism.
Roughly speaking, if some nonfaulty party decides on some value, there exist $f+1$ nonfaulty parties that are locked on that value and will stop any other value from progressing past the $propose$ message.
The liveness lemmas show two crucial properties for liveness.
First of all, if some nonfaulty party sets $key3$ to be some value, then there are $f+1$ parties that will support that key.
This means that if a nonfaulty party hears key suggestions from all nonfaulty parties, it accepts them and picks some key.
Secondly, if some nonfaulty primary picks a key to propose, the $suggest$ messages it receives guarantee that any nonfaulty party will receive enough supporting $proof$ messages.
This means that all nonfaulty parties eventually accept the primary's proposal, even if they are locked on some other value.
In the following lemmas assume that the number of faulty parties is $f<\frac{n}{3}$.

\subsection{Safety Lemmas}

The following lemma and corollary show that a primary cannot equivocate in a given view.
More precisely, in a given view all nonfaulty parties send messages that report the same value, other than $echo$ messages which might have more than one value.

\begin{restatable}{lemma}{nonEquivocation}\label{lem:nonEquivocation}
    If two nonfaulty parties send the messages $\left<key1,val,v\right>$ and $\left<key1,val',v\right>$, then $val=val'$.
\end{restatable}
\begin{proof}
    Observe two nonfaulty parties $i$ and $j$ that sent the messages $\left<key1,val,v\right>$ and $\left<key1,val',v\right>$ respectively.
    Before doing so, $i$ received the message $\left<echo,val,view\right>$ from $n-f$ parties, $f+1$ of which are nonfaulty.
    Let the set of those nonfaulty parties be $I$.
    Similarly, $j$ received the message $\left<echo,val',view\right>$ from $n-f$ parties.
    Since there are only $n$ parties, $j$ must have received the message from at least one party $k\in I$.
    Party $k$ is nonfaulty, so it sends the same $key1$ message to all parties.
    Therefore, $val=val'$.
\end{proof}

\begin{restatable}{corollary}{singleValue}\label{col:singleValue}
    If two nonfaulty parties $i$ and $j$ send a $\left<tag,val,v\right>$ and $\left<tag',val',v\right>$ message such that $tag,tag'\in\left\{key1,key2,key3,lock\right\}$ then $val=val'$.
\end{restatable}
\begin{proof}
    If some nonfaulty party sent a $\left<key1,val,v\right>$ message, from Lemma~\ref{lem:nonEquivocation} all nonfaulty parties that send a $\left<key1,val',v\right>$ message do so with $val=val'$.
    Now observe some nonfaulty party that sends a $\left<key2,val',v\right>$ message.
    It first received a $\left<key1,val',v\right>$ from $n-f$ parties, $f+1$ of which are nonfaulty.
    Using the previous observation, $val'=val$.
    Note that this also means that if two nonfaulty parties send $\left<key2,val,v\right>$ and $\left<key2,val',v\right>$ messages, then $val=val'$.
    Using similar logic, the same can be said for $key3$ and $lock$ messages.
\end{proof}

The following lemma and corollary now show that all $done$ messages that nonfaulty parties send have the same value.
There are two ways nonfaulty party might send a $done$ message: in the end of a view, or after receiving enough $done$ messages from other parties.
In the first view a nonfaulty party sends a $done$ message in line~\ref{line:doneMessage}, no nonfaulty party sends a $done$ message with another value because of the previous non-equivocation claims.
Then, once such a $done$ message is sent, there are $f+1$ nonfaulty parties that are locked on that value, and won't allow any other value to be proposed by a primary.
Since all nonfaulty parties send $done$ messages with the same value at the end of views, they never receive enough $done$ messages with another value for them to echo that $done$ message.

\begin{restatable}{lemma}{uniqueDone}\label{lem:uniqueDone}
    If two nonfaulty parties send the messages $\left<done,val\right>$ and $\left<done,val'\right>$ in line~\ref{line:doneMessage}, then $val=val'$.
\end{restatable}
\begin{proof}
    Let $v^*$ be the first view in which some nonfaulty party sends a $done$ message in line~\ref{line:doneMessage}.
    We will now prove by induction that: 
    \begin{itemize}
        \item for every $v\geq v^*$, no nonfaulty party sends a $\left<key1,val',v\right>$  message with $val'\neq val$,
        \item from that point on there exist $f+1$ nonfaulty parties with $lock\geq v^*,lock\_val=val$,
        \item no nonfaulty party has $prev\_key1\geq v^*$,
        \item and for every nonfaulty party that has $key1\geq v^*$, $key1\_val=val$.
    \end{itemize}
    First, observe the view $v^*$.
    Some nonfaulty party sent the message $\left<done,val\right>$ in line~\ref{line:doneMessage} during view $v$, and thus it received a $\left<lock,val,v^*\right>$ message from $n-f$ parties.
    Since $n-2f\geq f+1$, $f+1$ of those parties are nonfaulty.
    After sending such a message, every one of those nonfaulty parties set $lock= v^*,lock\_val=val$, which proves the second condition holds.
    From Corollary~\ref{col:singleValue}, every nonfaulty party that sends a $\left<key1,val',v^*\right>$ message does so with $val=val'$.
    Clearly before view $v^*$, all nonfaulty parties have $prev\_key1<key1<v^*$.
    This means that for every nonfaulty party that doesn't update its $key1$ or $prev\_key1$ fields in view $v^*$, the claim holds in the end of view $v^*$.
    Now observe some nonfaulty party that updates its $key1$ field in view $v^*$.
    If in the beginning of the view $key1\_val=val$, the party only updates $key1=v^*$, and thus the claim holds.
    Otherwise, $key\_val\neq val$, and the party first updates $prev\_key1=key1<v^*$ and then updates $key1=v^*,key\_val=val$.
    The claim holds in this case as well.
    
    Now, assume the claim holds for every $v'$ such that $v>v'\geq v^*$. 
    Let $I$ be a set of $f+1$ nonfaulty parties for whom $lock\geq v^*,lock\_val=val$.
    Observe some nonfaulty party $i$ that sends a $\left<key1,val',v\right>$ message.
    It must have first received an $\left<echo,val',v\right>$ message from $n-f$ parties.
    Since $\left|I\right|= f+1$, one of those parties is in the set $I$.
    Let that party be $j$.
    Party $j$ only sends such a message if either $val=val'$ or if it receives $proof$ messages that support opening its lock in view $v$ from $f+1$ parties, one of which is nonfaulty.
    However, for every nonfaulty party, $prev\_key1<v^*\leq lock$, so the first condition of Algorithm~\ref{alg:checkKey1} doesn't hold.
    If $key1\geq lock\geq v^*$ then $key1\_val=val=lock\_val$, which means that the second condition of Algorithm~\ref{alg:checkKey1} doesn't hold.
    Therefore, it must be the case that $val'=val$.
    In other words, if some nonfaulty party sends a $\left<key,val',v\right>$, $val'=val$, and thus the first part of the induction claim holds.
    
    If some nonfaulty party in $I$ updates its $lock$ value in view $v$, then it first received $\left<lock,val',v\right>$ messages from $n-f$ parties, $f+1$ of which are nonfaulty.
    From Corollary~\ref{col:singleValue}, if some nonfaulty party sends a $lock$ message in view $v$ it does so with $val'=val$.
    Therefore, if some nonfaulty party in $I$ updates its $lock$ value in view $v$, it sets $lock= v> v^*$ and $lock\_value= val$.
    In other words, $I$ remains a set of $f+1$ nonfaulty parties with $lock\geq v^*,lock\_val=val$.
    Now observe some nonfaulty party $i$.
    If $i$ doesn't update its $key1$ field, then it doesn't update its $prev\_key1$ field either.
    In that case, the final two conditions of the induction hold.
    If $i$ updates its $key1$ field, it must be the case that $i$ received $\left<key1,val,v\right>$ messages from $n-f$ parties.
    If $i$ already had $key1\_val=val$, then it only updates $key1=v>v^*$, and $prev\_key1<v^*$ remains true.
    If $i$ had $key1\_val\neq val$, then from the induction hypothesis it must be the case that $key1<v^*$.
    In that case, $i$ updates $prev\_key1=key1<v^*$ and then $key1=v>v^*,key\_val=val$.
    In both of those scenarios the claim continues to hold, completing the induction.
    
    Finally, if some nonfaulty party sends a $\left<done,val'\right>$ message in view $v\geq v^*$ in line~\ref{line:doneMessage}, then it first received a $\left<lock,val',v\right>$ message from $n-f$ parties, $f+1$ of which are nonfaulty.
    As shown above, in round $v$ at least $n-f$ parties sent some $\left<key1,val'',v\right>$ message, $f+1$ of which are nonfaulty as well.
    Combining the previous observations and Corollary~\ref{col:singleValue}, $val'=val''=val$, completing our proof.
\end{proof}

\begin{restatable}{corollary}{doneCorollary}\label{col:uniqueDone}
    If two nonfaulty parties send the messages $\left<done,val\right>$ and $\left<done,val'\right>$, then $val=val'$.
\end{restatable}
\begin{proof}
    Assume by way of contradiction that there exist two values $val\neq val'$ such that two nonfaulty parties sent the messages $\left<done,val\right>$ and $\left<done,val'\right>$.
    Let $i$ be the first nonfaulty party that sent a $\left<done,val\right>$ message and $j$ be the first nonfaulty party that sent a $\left<done,val'\right>$ message.
    If $i$ sent the message in line~\ref{line:doneEcho} it must have first received the message $\left<done,val\right>$ from $f+1$ parties, one of which is nonfaulty.
    By assumption, $i$ is the first nonfaulty party that sends such a message, reaching a contradiction.
    Similarly $j$ couldn't have sent the message in line~\ref{line:doneEcho}.
    This means that both of them sent their respective messages in line~\ref{line:doneMessage}, and thus from Lemma~\ref{lem:uniqueDone}, $val=val'$.
\end{proof}

\subsection{Liveness Lemmas}

The first two lemmas show that no nonfaulty party gets ``stuck'' in a view.
If some nonfaulty party terminates, then every nonfaulty party eventually terminates as well.
In addition, after GST, all nonfaulty parties start participating in consecutive views until terminating.

\begin{restatable}{lemma}{terminationPropagation}\label{lem:terminationPropagation}
    Observe some  nonfaulty party $i$ that terminates.
    All nonfaulty parties terminate no later than $2\Delta$ time after both GST occurs, and $i$ terminates.
\end{restatable}
\begin{proof}
    If $i$ terminates, then it received $\left<done,val\right>$ messages with the same value $val$ from $n-f$ parties, $f+1$ of which are nonfaulty.
    Every one of those nonfaulty parties sends the message to all nonfaulty parties.
    After GST, all parties receive those messages in $\Delta$ time or less.
    If when a nonfaulty party receives those messages it has already sent a $\left<done,val'\right>$ message, from Corollary~\ref{col:uniqueDone} it did so with $val'=val$.
    Otherwise, once it receives those messages it sends a $\left<done,val\right>$ message in line~\ref{line:doneEcho}.
    In other words, all nonfaulty parties eventually send a $\left<done,val\right>$ message.
    All nonfaulty parties then receive those messages in $\Delta$ more time.
    After receiving those messages, every nonfaulty party that hasn't terminated yet does so.
    In total, all of those messages arrive in no more than $2\Delta$ time, and then all nonfaulty parties terminate.
\end{proof}

\begin{restatable}{lemma}{abortPropagation}\label{lem:abortPropagation}
    Let $v$ be the highest view that some nonfaulty party is in at GST.
    For every view $v'>v$, all nonfaulty parties either start view $v'$, or terminate in some earlier view.
    
    Furthermore, if some nonfaulty party starts view $v'$ after GST, all nonfaulty parties either terminate or start view $v'$ no later than $2\Delta$ time afterwards.
\end{restatable}
\begin{proof}
    If some nonfaulty party terminated before GST, from Lemma~\ref{lem:terminationPropagation} all nonfaulty parties terminate $2\Delta$ time after GST.
    At GST, no nonfaulty party is in any view greater than $v$.
    This means that no nonfaulty party sends an $\left<abort,v'\right>$ message for any $v'>v$ in line~\ref{line:abortMessage} before GST.
    Therefore, at that time, for every nonfaulty party, the $f+1$'th largest value in $highest\_abort$ is no greater than $v$, and thus no nonfaulty party sends an $\left<abort,v'\right>$ message for any $v'>v$ in line~\ref{line:abortEcho} either.
    Since some nonfaulty party started view $v$, it found that the $n-f$'th largest value in $highest\_abort$ is $v-1$.
    In order for that to happen, it must have received $\left<abort,v_j\right>$ messages from $n-f$ parties $j$ such that $v_j\geq v-1$, and updated $highest\_abort$ accordingly.
    Out of those $n-f$ parties, at least $f+1$ are nonfaulty so they sent the same message to all parties and the nonfaulty parties receive that message in $\Delta$ time.
    Every nonfaulty party that receives those $f+1$ messages updates $highest\_abort$ accordingly and sees that the $f+1$'th largest value in $highest\_abort$ is at least $v-1$ and thus sends an $\left<abort,u\right>$ message for some $u\geq v-1$ to all parties.
    Every nonfaulty party receives those messages in $\Delta$ time, updates $highest\_abort$ accordingly and finds that the $n-f$'th largest value in $highest\_abort$ is at least $v-1$.
    Afterwards, they all compute the $n-f$'th largest value in $highest\_abort$, $w\geq v-1$, and start view $w+1$.
    Note that as discussed above, none of the messages those nonfaulty parties send is sent in a view greater than $v$, so all nonfaulty parties either start view $v$ or view $v+1$ in $2\Delta$ time.
    Using similar arguments, if some nonfaulty party starts view $v'$ after GST, all nonfaulty parties start view $v'$ as well in no more than $2\Delta$ time.
    Finally, after \viewLength time in each view, every nonfaulty party that hasn't terminated sends an $abort$ message, all nonfaulty parties receive that message in $\Delta$ time and start the next view until some nonfaulty party terminates.
    Crucially, no nonfaulty node sends an $\left<abort,v'\right>$ message in line~\ref{line:abortMessage} before \viewLength time has passed, so using similar arguments every nonfaulty hears all of the relevant $abort$ messages and starts view $v'$ before starting any later views.
\end{proof}

Eventually, all nonfaulty parties participate in some view with a nonfaulty primary, if they haven't terminated previously.
The next lemmas show that once that happens, all nonfaulty parties terminate.
First of all, in order for that to happen, a primary needs to receive enough suggestions for a $key3$ that it will accept.
The following lemma shows that every nonfaulty party's $key3$ field has enough support from nonfaulty parties for the primary to accept the key.
Intuitively, since a nonfaulty party set its $key3$ field to some value, there exist $f+1$ nonfaulty parties that sent a $key2$ message with that value.
The lemma shows that those $f+1$ nonfaulty parties have $key2$, $key2\_val$ and $prev\_key2$ fields that continue to support the key.

\begin{lemma}\label{lem:goodKeySupport}
    If some nonfaulty party sets $key3=v$, $key3\_val=val$ in view $v$, then there exist $f+1$ nonfaulty parties whose $suggest$ messages in every view $v'>v$ support $key3$ and $key3\_val$.
\end{lemma}
\begin{proof}
    We will prove by induction that there exist $f+1$ nonfaulty parties for whom in every $v'> v$ either $prev\_key2\geq key3$, or $key2\geq key$ and $key2\_value=val$.
    Since those are the fields that nonfaulty parties send in $suggest$ messages, that proves the lemma.
    First, observe view $v$.
    In that view, some nonfaulty party set $key3=v$ and $key3\_val=val$.
    This means that it received a $\left<key2,val,v\right>$ message from $n-f$ parties, $f+1$ of whom are nonfaulty.
    In addition to other possible updates, every one of those parties updates $key2=view$, and $key2\_val=val$ if that isn't true already.
    Those $f+1$ parties prove the claim for view $v$.
    
    Now assume the claim holds for every $v''<v'$.
    Observe party $j$, which is one of the $f+1$ parties described in the induction claim.
    If $j$ doesn't update any of its $key2$ fields in view $v'$, those conditions continue to hold in the end of view $v'$ and in the beginning of the next view.
    If $j$ only updates $key2$ to be $v'$, then if $prev\_key2\geq key3$, it remains that way, and if $key2\geq key3$ as well as $key2\_val=key3\_val$, after updating $key2$ to be $v'>key2\geq key3$, it also remains that way.
    Otherwise $j$ updates $prev\_key2=key2$ too.
    Note that $key2>prev\_key2$ at all times. Therefore, before updating $prev\_key2$, regardless of which part of the induction claim holds, $key2\geq key3$.
    After updating $prev\_key2$ to be $key2$, $prev\_key2\geq key3$, completing the proof.
\end{proof}

The following lemma is used to show that if a nonfaulty primary chose some key, and some nonfaulty party has a lock, it is either the case that the key's value equals the lock's value, or there are enough nonfaulty parties that support opening the lock.
Note that the conditions of the lemma are nearly identical to the conditions the primary checks before accepting a proof as supporting some key.
This means that before accepting a key, the primary essentially checks if there is enough support to open any other lock.
Similarly to the previous lemma, this lemma shows that if some nonfaulty party sets $key2$ to some value, there are $f+1$ parties that sent a $key1$ message with that value.
Those $f+1$ parties' $key1$, $key1\_val$ and $prev\_key1$ fields then continue to support any lock set previously with another value.

\begin{restatable}{lemma}{openLockSupport}\label{lem:openLockSupport}
    Let $lock>0$ be some nonfaulty party's lock and $lock\_val$ be its value.
    If some nonfaulty party either has $prev\_key2\geq lock$ or $key2\geq lock$ and $key2\_val\neq lock\_val$, then there exist $f+1$ nonfaulty parties whose $key1$, $key1\_val$ and $prev\_key1$ fields support opening the lock.
\end{restatable}
\begin{proof}
    Let $i$ be a nonfaulty party such that either $prev\_key2\geq lock$ or $key2\geq lock$ and $key2\_val\neq lock\_val$.
    If $key2\geq lock>0$ and $key2\_val\neq lock\_val$, $i$ received a $\left<key1,key1\_val,key2\right>$ message from $n-f$ parties in view $key2$.
    Out of those $n-f$ parties, at least $f+1$ are nonfaulty.
    On the other hand, if $prev\_key2\geq lock>0$, then for some pair of values $val,val'$ such that $val\neq val'$, $i$ received a $\left<key1,val,prev\_key2\right>$ message from $f+1$ nonfaulty parties in view $prev\_key2$ and a $\left<key1,val',key2\right>$ message from $f+1$ nonfaulty parties in view $key2>prev\_key2\geq lock$.
    At least one of the values $val,val'$ must not equal $lock\_val$ because $val\neq val'$.
    In other words, in both cases there exist $f+1$ nonfaulty parties that sent a $\left<key1,val,v\right>$ in view $v$ such that $val\neq lock\_val$ and $v\geq lock$.
    Let $I$ be the set of those nonfaulty parties.
    
    We now prove by induction that for every $v'\geq v$, all of the parties in $I$ either have $prev\_key1\geq lock$ or $key1\geq lock$ and $key1\_val\neq lock\_val$.
    First, observe view $v$. As stated above, in view $v$ all of the parties in $I$ sent a $\left<key1,val,v\right>$ and thus set $key1=v\geq lock$ and $key1\_val=val\neq lock\_val$, if it wasn't already so.
    Now, assume the claim holds for all views $v''<v'$.
    Note that the values of $key1$ and $prev\_key1$ only grow throughout the run.
    This means that if $prev\_key1\geq lock$ in the beginning of view $v'$, this will also be true at the end of view $v'$.
    On the other hand, if that is not the case, then in the beginning of view $v'$, $key1\geq lock$ and $key1\_val\neq lock\_val$.
    If the value of $key1\_val$ isn't updated in view $v'$, then $key1$ can only grow and thus the claim continues to hold.
    On the other, if the value of $key1\_val$ is updated in view $v'$, then for some $val'\neq val$ the following updates take place: $key1\gets v',key1\_val\gets val',prev\_key1\gets key1$.
    By assumption, in the beginning of view $v'$, $key1\geq lock$, and thus after the update $prev\_key1\geq lock$, completing the proof.
\end{proof}

This final lemma ties the two previous lemmas together.
Once a nonfaulty party is chosen as primary after GST, the primary receives enough keys, and each one of them has enough support to be accepted.
Then, after the key is sent, every nonfaulty party either has a lock with the same value, or there is enough support to open its lock.
From this point on, the view progresses easily and all nonfaulty parties terminate.

\begin{restatable}{lemma}{firstGoodPrimary}\label{lem:firstGoodPrimary}
    Let $v$ be the first view with a nonfaulty primary that starts after GST\footnote{More precisely, by ``starting after GST'', we mean that the first time some nonfaulty party has $view\geq v$ is after GST.}.
    All nonfaulty parties decide on a value and terminate in view $v$, if they haven't done so earlier. 
    
    Furthermore, if all messages between nonfaulty parties are actually delayed only $\delta$ time until being received, they decide on a value and terminate in $O\left(\delta\right)$ time.
\end{restatable}
\begin{proof}
    If some nonfaulty party terminated previously, from Lemma~\ref{lem:terminationPropagation}, all nonfaulty parties terminate in $2\Delta$ time after GST.
    Now assume no nonfaulty party previously terminated.
    Since $v$ starts only after GST, it is larger than the largest view some nonfaulty party is in at GST.
    In that case, from Lemma~\ref{lem:abortPropagation} every nonfaulty party eventually starts view $v$.
    Furthermore, from the second part of Lemma~\ref{lem:abortPropagation} they all start view $v$ no later than $2\Delta$ time after the first nonfaulty party that starts it.
    
    In the beginning of view $v$ all nonfaulty parties compute a nonfaulty primary $i$ for view $v$.
    Afterwards, the parties call $view\_change\left(v,i\right)$ and send a $\left<request,v\right>$ message to all nonfaulty parties.
    As stated above, no nonfaulty party started any view $v'>v$ yet, so no nonfaulty party received a $\left<request,v'\right>$ message for any $v'>v$ from any other nonfaulty party.
    After $\Delta$ time, the $request$ message sent by every nonfaulty party $j$ is received by all nonfaulty parties, and then they update $highest\_request\left[j\right]=v$.
    At that point, in every $send\_all\_upon\_join$ call, every nonfaulty party sends every other nonfaulty party the relevant message.
    Also note that no nonfaulty party sends any of the messages from its current view before receiving a $request$ message.
    Since the $request$ messages are sent only after starting the view, no relevant message sent by a nonfaulty party is discarded in this view.
    After receiving the primary's $request$ message, all nonfaulty parties send $i$ a $suggest$ message consisting of their $key3$ and $key3\_val$ fields, as well as $key2$, $key2\_val$ and $prev\_key2$ fields.
    In addition, they all send a $proof$ message message to all parties.
    Observe one $suggest$ message sent by party $j$.
    The primary $i$ receives that message in $\Delta$ time at most.
    If $key3=0$, $i$ adds $(key,key\_val)$ to $suggestions$.
    Note that whenever a nonfaulty party updates $key2$ and $prev\_key2$ it sets $prev\_key2<key2$.
    This means that $i$ adds $\left(key2,key2\_val,prev\_key2\right)$ to $key2\_proofs$.
    From Lemma~\ref{lem:goodKeySupport}, every nonfaulty party's $key3$, $key3\_val$ pair has $f+1$ supporting $suggest$ messages from nonfaulty parties, which means that $i$ eventually adds the tuple $\left(key3,key3\_val\right)$ to $suggestions$ for each message received from a nonfaulty party.
    Then, $i$ sees that $\left|suggestions\right|\geq n-f$, picks some tuple $\left(key,val\right)\in suggestions$ with a maximal $key$, and sends the message $\left<propose,key,val,v\right>$ to all parties.
    
    All nonfaulty parties receive that $propose$ message from the primary in $\Delta$ more time, as well as the $proof$ message sent by all nonfaulty parties.
    Now, observe some locked party $l$.
    If $l$ sees that $lock=0$ or $val=lock\_val$, it sends an $\left<echo,val,v\right>$ message to all nonfaulty parties.
    If that is not the case, $l$ has set $lock=v'>0$ in some view $v'<v$, and $lock\_val=val'$ for some $val'\neq val$.
    Before setting its lock in view $v'$, $l$ must have received a $\left<key3,val',v'\right>$ from $n-f$ parties.
    Out of those parties, $f+1$ are nonfaulty, and they set they $key3$ field to be $v'$.
    Note that nonfaulty parties only increase their $key3$ field throughout the protocol.
    This means that out of the $n-f$ suggestions received by $i$, at least one was for a $key3$ that is greater than $v'$.
    Since $i$ takes the tuple with the maximal key, we know that $key\geq v'$.
    
    Note that since $\left(key,val\right)\in suggestions$, if $key\neq 0$, $i$ received $f+1$ $suggest$ messages that support the pair of values. 
    At least one of those messages was sent by a nonfaulty party $j$.
    Observe $j$'s message supporting the pair $key$, $val$.
    The message contains $j$'s $key2$, $key2\_val$ and $prev\_key2$ fields.
    The message supports the pair $key$ and $val$ and thus it is either the case that $prev\_key2\geq key\geq lock$ or that $key2\geq lock$ and $key2\_val=val\neq lock\_val$.
    From Lemma~\ref{lem:openLockSupport}, there are $f+1$ parties whose $key1$, $key1\_val$ and $prev\_key1$ fields support opening $l$'s lock, and thus they send $proof$ messages that support opening the lock.
    Party $l$ receives those messages and then sends an $\left<echo,val,v\right>$ message as well.
    
    All nonfaulty parties receive the $echo$ message from $n-f$ parties $\Delta$ time later and then send a $\left<key1,val,v\right>$ message to all parties.
    Similarly, after $\Delta$ more time all nonfaulty parties receive a $\left<key1,val,v\right>$ message from all nonfaulty parties and send $\left<key2,val,v\right>$ message.
    The same argument can be made for sending a $key3$ message, a $lock$ message, and a $done$ message, each taking $\Delta$ more time.
    Finally, after $\Delta$ more time, all parties receive the $done$ messages.
    It is important to note that if some nonfaulty party has already sent a $\left<done,val',v'\right>$ message in view $v'<v$, then from Corollary~\ref{col:uniqueDone}, it does so with $val=val'$.
    After receiving those $done$ messages from $n-f$ parties, all nonfaulty parties decide on a value and terminate.
    
    The overall time for all of the steps is $2\Delta+9\cdot\Delta=\viewLength$.
    During that time, no nonfaulty party sends an $\left<abort,v\right>$ message, so all nonfaulty parties reach the point in which they decide on $val$ without having to change view.
    If all of those messages are actually delivered in $\delta$ time or less, then the previous calculations can be done with $\delta$ instead, finding that all nonfaulty decide on a value in $11\delta=O\left(\delta\right)$ time instead.
\end{proof}

\subsection{Main Theorem}

Using the previous lemmas, it is now possible to prove the main theorem:
\begin{theorem}
    Algorithm~\ref{alg:BA} is a Byzantine Agreement protocol in partial synchrony resilient to $f<\frac{n}{3}$ Byzantine parties.
\end{theorem}
\begin{proof}
    We prove each property individually.
    
    \textbf{Correctness.} Observe two nonfaulty parties $i,j$ that decide on the values $val,val'$ respectively.
    Party $i$ first received a $\left<done,val\right>$ message from $n-f$ parties, and $j$ received a $\left<done,val'\right>$ message from $n-f$ parties.
    Since $n-f>f$, $i$ and $j$ receive at least one of their respective messages from some nonfaulty party.
    From Corollary~\ref{col:uniqueDone}, all nonfaulty parties that send a $done$ message do so with the same value.
    Therefore, $val=val'$.

    \textbf{Validity.} Assume that all parties are nonfaulty and that they have the same input $val$.
    We will prove by induction that for every view $v$, every nonfaulty party has $key3\_val=val$.
    Furthermore, if some nonfaulty party sends a $\left<key1,val',v\right>$ message, then $val'=val$.
    First, all parties set $key3\_val$ to be $val$ in the beginning of the protocol.
    Assume the claim holds for every $v'<v$.
    In the beginning of view $v$, the primary calls the $view\_change$ protocol.
    Before completing $view\_change$, the primary receives $suggest$ messages from $n-f$ parties with their $key3\_val$ field.
    Since all parties are nonfaulty, they all send the $key3\_val$ they have at that point, and from the induction hypothesis $key3\_val=val$.
    This means that if the primary completes the $view\_change$ protocol, it sees that for every $\left(key,key\_val\right)\in suggestions$, $key\_val=val$ and thus if the primary sends a $propose$ message it sends the message $\left<propose,val,key,v\right>$ to all parties.
    Now, every nonfaulty party that sends a $key1$ message sends the message $\left<key1,val,v\right>$.
    From Corollary~\ref{col:singleValue}, every nonfaulty party that sends a $\left<key3,val',v\right>$ message, does so with $val'=val$.
    If a nonfaulty party updates $key3\_val$ to a new value $val'$, it also sends a $\left<key3,val',v\right>$ message.
    However, as shown above the only value sent in such a message is $val$ so no nonfaulty party updates its $key3\_val$ field to any other value.
    Using Corollary~\ref{col:singleValue}, every nonfaulty party that sends a $lock$ message does so with the value $val$.
    This means that any party that sends a $done$ message in line~\ref{line:doneMessage}, does so with the value $val$.
    Clearly any party that sends a $done$ message in line~\ref{line:doneEcho} does so with the value $val$ as well, because it never receives $done$ messages with any other value.
    Finally, this means that every nonfaulty party that decides on a value decides on $val$.
    
    \textbf{Termination.} 
    Observe the system after GST, and let $v$ be the highest view that some nonfaulty party is in at that time.
    From Lemma~\ref{lem:abortPropagation}, all nonfaulty parties either terminate or participate in every view $v'>v$.
    Since the primaries are chosen in a round-robin fashion, after no more than $f+1$ views, some nonfaulty party starts a view with a nonfaulty primary.
    From Lemma~\ref{lem:firstGoodPrimary}, all nonfaulty parties either terminate in that view or earlier.
\end{proof}

\subsection{Complexity Measures}

The main complexity measures of interest are round complexity, word complexity, and space complexity.

\textbf{Word complexity.} In IT-HS in every round, every party sends at most $O(1)$ words to every other party. We assume that a \emph{word} is large enough to contain any counter or identifier. This implies that just $O(n^2)$ words are sent in each round. 

\textbf{Round complexity.} As IT-HS is a primary-backup view-based protocol (like Paxos and PBFT), there are no bounds on the number of rounds while the system is still asynchronous.
Therefore, we use the standard measure of counting the number of rounds and number of words sent after GST.
Furthermore, in order to be useful in the task of agreeing on many values, a desirable property is \emph{optimistic responsiveness}: when the primary is nonfaulty and the network delay is low, all nonfaulty parties complete the protocol at network speed.
This desire is captured in the next definition:
\begin{definition}[Optimistic Responsivness]
    Assume all messages between nonfaulty parties are actually delivered in $\delta <\Delta$ time.
    The protocol is said to be \emph{optimistically responsive} if all nonfaulty parties complete the protocol in $O\left(\delta\right)$ time after a nonfaulty primary is chosen after GST.
\end{definition}

\textbf{Space complexity}. We separate the local space complexity into two types: persistent memory and transient memory.  
In this setting, parties can crash and be rebooted.
Persistent memory is never erased, even in the event of a crash, while transient memory can be erased by a reboot event. 
IT-HS requires asymptotically optimal $O(1)$ persistent storage (measured in words) and just $O(1)$ transient memory per communication channel (so a total of $O(n)$ transient memory).

After a reboot, nonfaulty parties can ask other parties to send messages that help recover information needed in their transient memory.
In this setting we assume that all nonfaulty parties that terminate still reply to messages asking for previously sent information.

\begin{theorem}\label{thm:complexity}
    During Algorithm~\ref{alg:BA}, each nonfaulty party sends a constant number of words to each other party in each view and requires $O\left(n\right)$ memory overall, out of which $O(1)$ is persistent memory.
    Furthermore, the protocol is optimistically responsive.
\end{theorem}
\begin{proof}
    First note that each view consists of one message sent from all parties to the primary, one message sent from the primary to all parties, and a constant number of all-to-all communication rounds.
    In addition, each message consists of no more than $7$ words.
    Overall, each party only sends a constant number of messages to every party, each with a constant number of words.
    In each view, every nonfaulty party needs to remember which messages were sent to it by other parties, as well as a constant amount of information about every $suggest$ and $proof$ message.
    Since a constant number of words and messages is sent from each party to every other party, this requires $O(n)$ memory.
    Note that once a new view is started, all of the information stored in the previous call to $view\_change$ and $process\_messages$ is freed.
    Other than that, every nonfaulty party allocates two arrays of size $n$, a constant number of other fields, and needs to remember the first $done$ messages received from every other party.
    This also requires $O(n)$ memory.
    Overall, the only fields that need to be stored in persistent memory are the $view$, $lock$, $lock\_val$, and various $key$, $key\_val$ and $prev\_key$ fields, as well as the messages it sent in the current view, and the last $done$, $request$ and $abort$ messages it sent.
    This is a constant number of fields, in addition to a constant number of messages.
    After being rebooted, a nonfaulty party $i$ can ask to receive the last $done$, $request$, and $abort$ messages sent by all nonfaulty parties to restore the information it lost that doesn't pertain to any specific view, and any message sent in the current view.
    In addition, it sends a $request$ message for its current view.
    Upon receiving such a message, a nonfaulty party $j$ replies with the last $done$, $request$ and $abort$ messages it sent.
    In addition, if $j$ is in the view that party $i$ asked about, it also re-sends the messages it sent in the current view.
    Note that this is essentially the same as $i$ receiving messages late and starting its view after being rebooted, and thus all of the properties still hold.
    The fact that the protocol is optimistically responsive is proven in Lemma~\ref{lem:firstGoodPrimary}.
\end{proof}

\section{Multi-Shot Byzantine Agreement and State Machine Replication}

This section describes taking a Byzantine Agreement protocol and using it to solve two tasks that are natural extensions of a single shot agreement.
Both tasks deal with different formulations for the idea of agreeing on many values, instead of just one.

\subsection{State Machine Replication with Stable Leader (a la PBFT)}

In the task of State Machine Replication \cite{schneider1990SMR}, all parties (called replicas) have knowledge of the same state machine.
Each party receives a (possibly infinite) series of instructions to perform on the state machine as input.
The goal of the parties is to all perform the same actions on the state machine in the same order.
More precisely, the parties are actually only interested in the state of the state machine, and aren't required to see all of the intermediary states throughout computation.
In order to avoid trivial solutions, if all parties are nonfaulty and they have the same $s$'th instruction as input, then they all execute it as the $s$'th instruction for the state machine.
This task can be achieved utilizing any Byzantine Agreement protocol, using ideas from the PBFT protocol.

In addition to the inputs, the protocol is parameterized by a window size $\alpha$.  
All parties participate in $\alpha$ instances of the Byzantine Agreement protocol, each one tagged with the current decision number.
After each decision, every party saves a log of their current decision, and updates the state machine according to the decided upon instruction.
Then, after every $\frac{\alpha}{2}$ decisions, each party saves a ``checkpoint'' with the current state of the state machine, and deletes the log of the $\frac{\alpha}{2}$ oldest decisions.
Then, before starting the next $\frac{\alpha}{2}$ decisions, every party sends its current checkpoint and makes sure it receives the same state from $n-f$ parties using techniques similar to Bracha broadcast.
Furthermore, as long as no view fails, the primary isn't replaced.
This means that eventually at some point, either there exists a faulty primary that always acts like a nonfaulty primary, or a nonfaulty primary is chosen and is never replaced.
Both sending the checkpoints and replacing faulty leaders require more implementation details which can be found in \cite{castro2001PBFT}. 

Using these techniques, all parties can decide on $O\left(\alpha\right)$ instructions at a time, improving the throughput of the algorithm.
The communication complexity per view remains similar to the communication complexity of the IT-HS algorithm, but once a nonfaulty primary is reached after GST, all invocations of the protocol require only one view to terminate.
Alternatively, if a nonfaulty primary is never reached after GST, a faulty party acts like a nonfaulty primary indefinitely, which yields the same round complexity.
Finally, if we assume that a description of the state machine requires $O\left(S\right)$ space, the protocol now requires $O(S + \alpha)$ persistent space in order to store the checkpoints and store the $O(1)$ state for each slot in the window.
In addition, the protocol requires $O\left(\alpha\cdot n+S\right)$ transient space in order to store the information about all active calls to IT-HS, the $\alpha$ decisions in the log, and a description of the current state of the state machine.

\subsection{Multi-Shot Agreement with Pipelining (a la HotStuff)}
In contrast, we can take the approach of HotStuff \cite{abraham2019hotstuff} and solve the task of multi-shot agreement. 
In this task, party $i$ has an infinite series of inputs $x_i^1,x_i^2,\ldots$, and the goal of the parties is to agree on an infinite number of values.
Each decision is associated with a slot which is the number $s\in\mathbb{N}$ of the decision made.
Each one of these decisions is required to have the agreement properties, i.e.: eventually all nonfaulty parties decide on a value for slot $s$, they all decide on the same value, and if all parties are nonfaulty and have the same input $val$ for slot $s$, the decision for the slot is $val$.

A naive implementation for this task is to sequentially call separate instances of IT-HS for every slot $s\in\mathbb{N}$, each with the input $\left(s,x_i^s\right)$. In order to improve the throughput of the protocol, after completing an instance of the IT-HS protocol, the parties can continue with the next view and the next primary in the round-robin.
This slight adjustment ensures that after GST, $n-f$ out of every $n$ views have a decision made, and if messages between nonfaulty parties are only delayed $\delta$ time, each one of those views requires only $O(\delta)$ time to reach a decision.
Slight adjustments need to be made in that case so that $abort$ messages are sent about views regardless of the slot, so that all parties continue participating in the same views throughout the protocol.
In addition, messages about different slots need to be ignored.

In the case of the optimistic assumption that most parties are nonfaulty, a significantly more efficient alternative can be gleaned from the HotStuff protocol. 
This alternative uses a technique called \textit{pipelining} (or chaining).
Roughly speaking, in this technique, all parties start slot $s$ by appending messages, starting on the second round (round, not view) of slot $s-1$.
In the case of HT-IS, the protocol can be changed so that $suggest$ messages are sent to all parties, and then each party starts slot $s$ after receiving $n-f$ $suggest$ messages in slot $s$.
Note that the exact length of timeouts needs to be slightly adjusted, and the details can be found in \cite{abraham2019hotstuff}.
In slot $s$, a nonfaulty primary appends its current proposal to the proposal it heard in slot $s-1$.
Then, before deciding on a value in slot $s$, parties check that the decision values in the previous slots agree with the proposal in slot $s$.
If they do, then the parties agree on the value in this slot as well.
In this protocol, each view lasts for \viewLength time, so if at some point a primary sees that a proposal from $11$ views ago failed, it appends its proposal to the first one that it accepted from a previous view.
After GST, if there are $m+11$ nonfaulty primaries in a row, then the last $m$ primaries are guaranteed to complete the protocol, and thus add $m$ decisions in $(m+11)\Delta$ time instead of in $m\cdot\viewLength$ time.
This means that in the optimistic case that a vast majority of parties are nonfaulty, the throughput of this protocol is greatly improved as compared to the naive implementation.
In this protocol the communication complexity per view is still $O(n^2)$ messages, but a larger number of words.
However, note that it is not always the case that if a nonfaulty primary is chosen, its proposal is accepted.
To obtain bounded memory requirement one needs to add a checkpointing mechanism, similar to PBFT.  
As in PBFT,  only $O(n)$ transient space and $O(1)$ persistent space are required per decision in addition to the log of the decisions.

\bibliographystyle{plain}
\bibliography{bibfile}

\end{document}